\documentclass[12pt]{article}
\usepackage[letterpaper,margin=1in]{geometry} 

\usepackage{graphicx}
\usepackage{subfig}
\usepackage{color} 
\usepackage{url}
\usepackage{amsmath}
\usepackage{amsthm}
\usepackage{amssymb}
\usepackage{comment}
\usepackage{xcolor}

\numberwithin{equation}{section}
\newtheorem{theorem}{Theorem}[section]
\newtheorem*{theorem*}{Theorem}
\newtheorem{proposition}[theorem]{Proposition}
\newtheorem{lemma}[theorem]{Lemma}

\newtheorem{corollary}[theorem]{Corollary}
\theoremstyle{definition}
\newtheorem{remark}[theorem]{Remark}

\newtheorem{example}[theorem]{Example}
\newtheorem{definition}[theorem]{Definition}

\newcommand{\V}{\mathcal{V}}
\newcommand{\R}{\mathbb{R}}
\newcommand{\Z}{\mathbb{Z}}
\newcommand{\K}{\mathbb{K}}

\newcommand{\C}{\mathbb{C}}

\newcommand{\X}{X}
\newcommand{\Y}{Y}

\newcommand{\hatM}{\widehat{\M}}

\renewcommand{\O}{\operatorname{O}}
\newcommand{\br}[1]{{\left({#1}\right)}} 
\newcommand{\B}{\mathcal{B}}
\newcommand{\M}{\mathcal{M}}
\newcommand{\A}{\mathcal{A}}

\newcommand{\cM}{\mathcal{M}}

\newcommand{\rhog}{{d_{\textup{Gram}}}}
\newcommand{\rhos}{{d_{\sigma}}}

\newcommand{\rhoh}{{d_{\textup{H}}}}

\newcommand{\of}[1]{{\left({#1}\right)}} 
\newcommand{\norm}[1]{{\left\lVert{#1}\right\rVert}} 
\newcommand{\diag}{{\textup{diag}}}
\newcommand{\eqxspace}{{\hspace{0.0pt}}} 
\newcommand{\eqx}[1][]{{\eqxspace\overset{{\textup{#1}}}{=}\eqxspace}}
\newcommand{\leqx}[1][]{{\eqxspace\overset{{\textup{#1}}}{\leq}\eqxspace}}
\newcommand{\geqx}[1][]{{\eqxspace\overset{{\textup{#1}}}{\geq}\eqxspace}}

\DeclareMathOperator{\Aff}{Aff}
\DeclareMathOperator{\GL}{GL}

\DeclareMathOperator{\SO}{SO}

\newcommand{\rev}[1]{{\color{black}  #1}}

\newcommand{\twopartdef}[4]
{
	\left\{
		\begin{array}{ll}
			#1 & \mbox{if } #2 \\
			#3 & \mbox{if } #4
		\end{array}
	\right.
}

\begin{document}
	
	\title{The stability of generalized phase retrieval problem over compact groups}
	
	\author{Tal Amir, Tamir Bendory, Nadav Dym, and Dan Edidin}
	
    \date{}
	
	\maketitle

	\begin{abstract} 
The generalized phase retrieval problem over compact groups aims to recover a set of matrices—representing an unknown signal—from their associated Gram matrices. This framework generalizes the classical phase retrieval problem, which reconstructs a signal from the magnitudes of its Fourier transform, to a richer setting involving non-abelian compact groups. In this broader context, the unknown phases in Fourier space are replaced by unknown orthogonal matrices that arise from the action of a compact group on a finite-dimensional vector space.
This problem is primarily motivated by advances in electron microscopy to determining the 3D structure of biological macromolecules from highly noisy observations. To capture realistic assumptions from machine learning and signal processing, we model the signal as belonging to one of several broad structural families: a generic linear subspace, a sparse representation in a generic basis, the output of a generic ReLU neural network, or a generic low-dimensional manifold.
Our main result shows that, \rev{for a prior of sufficiently low dimension}, the generalized phase retrieval problem not only admits a unique solution (up to inherent group symmetries), but also satisfies a bi-Lipschitz property. This implies robustness to both noise and model mismatch---an essential requirement for practical use, especially when measurements are severely corrupted by noise.
These findings provide theoretical support for a wide class of scientific problems under modern structural assumptions, and they offer strong foundations for developing robust algorithms in high-noise regimes.

\end{abstract}

\section{Introduction}

\textbf{Phase retrieval.}
Phase retrieval is the problem of reconstructing a signal from the magnitudes of its Fourier transform. Since phase information is lost in Fourier space, accurate signal recovery requires incorporating additional prior knowledge about the signal, modeled by an appropriate ``signal space" $\Omega$.  
For a discrete signal \( x \), let \( X:= (X_1, \ldots, X_L) \) denote its Fourier coefficients,
and define the Fourier magnitudes \rev{(power spectrum)} as \( X^* X:= (X_1^* X_1, \ldots, X_L^* X_L) \). The phase retrieval problem can then be formulated as the task of finding a signal whose Fourier coefficients \( Z \) are consistent with the observed magnitudes and lie in the  signal space  \( \Omega \), i.e.,  
\begin{equation} \label{eq:pr}  
	\text{find } Z \in \Omega \quad \text{subject to} \quad Z^* Z \approx X^* X.
\end{equation}

The signal space \( \Omega \) depends on the specific application. The phase retrieval problem traces its origins to the early 20th century, emerging from advancements in X-ray crystallography, a fundamental technique for determining molecular structures that has driven major scientific advancements. In this context, \( \Omega \) consists of sparse signals, where the nonzero entries indicate the locations of atoms~\cite{elser2018benchmark}.  
In coherence diffraction imaging, the signal \( x \) is typically limited to a known support~\cite{shechtman2015phase,barnett2022geometry}.   
For more relevant applications, we refer the reader to recent surveys on the mathematical and numerical perspectives of the phase retrieval problem, along with the references therein~\cite{bendory2017fourier, grohs2020phase, fannjiang2020numerics}.  

\textbf{Multi-reference alignment.}
We now present an alternative perspective on the phase retrieval problem by leveraging the multi-reference alignment (MRA) model, a mathematical framework that encapsulates a broad class of estimation problems characterized by inherent group structures~\cite{bandeira2014multireference,bandeira2023estimation,perry2019sample,bendory2017bispectrum}. 
Let \( G \) be a compact group acting on a finite-dimensional, real vector space \( V \). Each MRA observation \( y \) is modeled as:  
\begin{eqnarray} \label{eq:mra}  
	y = g \cdot x + \varepsilon,  
\end{eqnarray}  
where \( g \in G \), \( \varepsilon\) is noise (or error) term independent of \( g \), \rev{the symbol} \( \cdot \) denotes the group action, and \( x \in V \) is the signal of interest.  We assume that \( g \) is uniformly distributed over \( G \) with respect to the Haar measure. The objective in the MRA model is to estimate the signal \( x \in V \) from \( n \) realizations:  
\begin{equation}  
	y_i = g_i \cdot x + \varepsilon_i, \quad i = 1, \ldots, n.  
\end{equation}

Recall that a general finite-dimensional representation of a compact group
$G$ can be decomposed as 
\begin{equation} \label{eq:V}
	V = \oplus_{\ell = 1}^L V_\ell^{\oplus R_\ell},   
\end{equation}
\rev{where} the $V_\ell$ are
non-isomorphic
irreducible representations of $G$ of dimension $N_\ell$. 
This means that an element $x \in V$ has a decomposition $x = \sum_{\ell = 1}^L \sum_{i= 1}^{R_\ell} x_\ell[i]$, 
where $x_\ell[i]$ 
is in the $i$-th copy of the irreducible representation $V_\ell$, and the $G$-action preserves this decomposition.
Conveniently, once a basis for each irreducible representation is fixed, an element of $V$ can be represented by an $L$-tuple
\begin{equation} \label{eq:signal}
	X:= (X_1, \ldots ,X_L)\in\R^{N_1 \times R_1}\times \cdots\times \R^{N_L \times R_L},
\end{equation}
where $X_\ell:=\{x_\ell[i]\}_{i=1}^{R_\ell}$ is an 
$N_\ell \times R_\ell$ matrix. In the sequel, we refer interchangeably to both $x\in V$ and the tuple~\eqref{eq:signal} as the signal.

\textbf{Second moment analysis.} Let us focus on the second moment of the MRA model. 
Based on the law of large numbers, the population second moment  can be approximated using the empirical second moment (the noise term is omitted for brevity):
\begin{equation}  
	\frac{1}{n} \sum_{i=1}^n y_i y_i^\top\approx\mathbb{E}[yy^\top] = \int_G (g \cdot x)(g \cdot x)^\top \, dg,  
\end{equation}
where the approximation is accurate when $n=\omega(\sigma^4)$, \rev{with $\sigma^2$ denoting the variance of the noise.}
In~\cite{bendory2024sample}, it was shown that the second moment of a  vector $x\in V$ is given by the $L$-tuple of the $R_\ell\times R_\ell$ symmetric matrices
\begin{equation}\label{eq:Gram}
	X^\top X:=(X_1^\top X_1,\ldots,X_L^\top X_L)\in\R^{R_1 \times R_1}\times \cdots\times \R^{R_L \times R_L}.    
\end{equation}
In particular, when $G=\Z_N$, the group of circular translations acting on $V=\R^N$, then the second moment results in the $L$-tuple~\eqref{eq:Gram}, where $X_\ell$ is the $\ell$-th Fourier coefficient, expanded in the real Fourier basis (of sines and cosines)~\cite{edidin2023generic} and $L = \left\lfloor \frac{N}{2} \right\rfloor +1$.
This is exactly the power spectrum mentioned above, and thus, the classical phase retrieval problem is a special case of the second moment of the MRA model \rev{in which} the irreducible representations are one- or two-dimensional and  multiplicity-free~\cite{bendory2023phase}.

\textbf{The generalized phase retrieval problem.} Recasting the phase retrieval problem from the perspective of second moments yields a natural generalization in which the cyclic group $\mathbb{Z}_L$ is replaced by an arbitrary compact group $G$ acting linearly on a vector space~$V$.
We refer to this problem as the 
\emph{generalized phase retrieval problem over compact  groups}~\cite{bendory2025generalized}.

\begin{definition}[The generalized phase retrieval problem over compact groups]
	Recover a signal $X$ of the form~\eqref{eq:signal} from the tuple of Gram matrices $X^\top X$~\eqref{eq:Gram}.
\end{definition}

Special cases of the generalized phase retrieval problem include the standard phase retrieval problem, and the principal motivation for this work: Determining the spatial structure of biological molecules using single-particle cryo-electron microscopy (cryo-EM). Section~\ref{sec:cryoEM} discusses this application in more detail.

The tuple of Gram matrices determines the signal only up to a set of orthogonal matrices, 
and this ambiguity can be alleviated using the fact that in many practical scenarios, prior structural information about the signal is available. In a previous
paper~\cite{bendory2024transversality}, we study a class of priors called {\em semi-algebraic sets}
which includes many priors commonly encountered in scientific and engineering applications. The main result
of~\cite{bendory2024transversality} gives conditions on these priors, ensuring that a signal is uniquely determined by its second moment (possibly up to \rev{a} global sign). We call this property \emph{transversality}. 

Transversality alone is insufficient for practical applications. In real-world scenarios, it is essential to guarantee stability---that is, small errors in the observations, which are inevitable in practice, should only cause small errors in signal recovery. To analyze stability, we adopt the powerful notion of bi-Lipschitz continuity. The goal of this paper is to show that on many typical real-world priors, the second moment is both transverse and
stable.
The main result of this paper is to extend the uniqueness results of~\cite{bendory2024transversality} and show that, under similar conditions, the generalized phase retrieval problem is in fact bi-Lipschitz.
This, in turn, implies that the map between the tuple of Gram matrices~\eqref{eq:Gram} and the signal~\eqref{eq:signal} is robust to noise and errors. The following theorem summarizes informally the main results of the paper, while the  
main theorems are given in Section~\ref{sec:Bi-Lipschitz-results}.  

\begin{theorem}[Main result, informal]
	Suppose that \(\cM\) is one of the following: (i) a generic linear space, (ii) a generic sparse prior, (iii) the image of a generic ReLU deep neural network~\eqref{eq:relu} or (iv) a compact manifold\rev{, each} of dimension \(M\). Then, for some constant \(C \leq 4\), if \(C M<K\), the mapping from the tuple of matrices in \eqref{eq:signal} to the corresponding Gram matrices in \eqref{eq:Gram} is bi-Lipschitz.
\end{theorem}

Here, the constant $K$ is the effective dimension of the representation under the action of the ambiguity group $H = \prod_{\ell =1}^L \O(N_\ell)$. For details, see Section~\ref{sec:transversality}.

\paragraph{Related literature.}
In a series of fundamental papers~\cite{bandeira2014saving, balan2015invertibility, balan2015lipschitz, balan2016lipschitz, balan2016reconstruction},  \rev{Bandeira, Cahill, Mixon and Nelson as well as} Balan, Wang and Zou 
proved that the frame phase retrieval problem is robust. 
In particular, if $\{f_1, \ldots, f_L\}$ is an $L$-element frame in $\K =\R^M$ or $\K=\C^M$
and if the phaseless measurement map 
$$  \K^M \to \R^L_{\geq 0};\,  x \mapsto (|\langle x,f_1 \rangle|, \ldots , |\langle x, f_L \rangle|),$$
is injective (modulo global phase), then it is also bi-Lipschitz. These results can be interpreted in our framework as stating
that if $X = (X_1, \ldots , X_L)$ is an $L$-tuple of $1 \times 1$ matrices and $\M \subset \R^L$ is the range of the analysis operator of the frame, then the map $$\M \to \R^L\, ; \, X \mapsto \sqrt{X^TX},$$ is bi-Lipschitz whenever
it is transverse to the orbits of $H = \O(1)^L = \{\pm 1\}^L$. In particular, Corollary~\ref{thm.linear_bilipschitz} is a natural extension of these works to Gram matrices, and also generalizes and extends a recent result of Derksen~\cite{derksen2024bi}, who proved
that the map $X \mapsto \sqrt{X^TX}$ induces a bi-Lipschitz map $\R^{n \times d}/\O(n) \to \R^{d \times d}$. 

The second moment we consider is an invariant for the group $H= \prod_{i=1}^L \O(N_L)$. There has been extensive recent
work by several authors proving bi-Lipschitz bounds for piecewise-linear group invariant functions for finite
groups. This includes max-filter~\cite{cahill2024towards,cahill2022group,qaddura2025max, mixon2023max} and the method of co-orbits~\cite{balan2023ginvariant, balan2022permutation, balan2023relationships, balan2024stability}, and other bi-Lipschitz constructions for permutation groups \cite{davidson2025on,sverdlov2024fsw,amir2025fourier,dym2025bi}. Notably, all these constructions discuss powerful invariants that determine the signal uniquely up to symmetry in the whole linear 'ambient' signal space. In this case, stability is relatively straightforward. In contrast, we focus here on simple invariants, \rev{namely the second moments}, where unique recovery and stability rely on more complex assumptions of non-linear priors.

\paragraph{Organization of the paper.}
The next section introduces the main background required for this work, including basic definitions and the statement of the transversality theorem proved in~\cite{bendory2024transversality}. 
Our main theoretical results on stability are presented in Section~\ref{sec:Bi-Lipschitz-results}, where we also explore their implications for several fundamental priors commonly used in machine learning and signal processing. The proof of the main theorem is provided in Section~\ref{sec:proof_meta_thm}.
We then introduce the computational aspects of cryo-EM and discuss how our results apply to this setting in Section~\ref{sec:cryoEM}. Finally, in Section~\ref{sec:outlook}, we outline several important open questions that remain in this field.

\section{The transversality theorem and bi-Lipshitz analysis for the second moment}

\subsection{Semi-algebraic sets} For a unique solution to the generalized phase retrieval problem, a structural assumption on the signal is necessary. Without such an assumption, the tuple of Gram matrices determines the signal only up to a set of orthogonal transformations. A broad and widely applicable assumption in engineering and scientific contexts is that the signal belongs to a low-dimensional semi-algebraic set, referred to as a \emph{semi-algebraic prior} on the signal.

A semi-algebraic set $\M \subset \R^N$ is a finite union of sets, which are defined by polynomial equality and inequality.
This work is motivated by three important special cases of semi-algebraic sets:

\begin{itemize}
	\item \emph{Linear priors.} The assumption that the signal lies in a low-dimensional subspace. Linear priors are ubiquitous in signal processing and machine learning and are the root of popular methods, such as Principal Component Analysis (PCA)\rev{, e.g.,~\cite{jolliffe2011principal,candes2011robust,castells2007principal}}. 
	
	\item \emph{Sparse priors.} The assumption that many signals can be approximated by only a few coefficients under some basis or frame \rev{, e.g., ~\cite{elad2010sparse,donoho2006compressed,zou2006sparse}}.
	
	\item \emph{ReLU neural networks.} These are based on neural networks of the form
	\begin{equation} \label{eq:relu}
		x=A_\ell\circ \eta\circ A_{\ell-1}\circ \ldots \circ \eta \circ A_1(z),
	\end{equation}
	where $z$ resides in a low-dimensional latent space, the $A_i$'s are affine transformations and  $\eta(u)=\max(u,0)$ is the element-wise rectification function (ReLU)\rev{~\cite{goodfellow2016deep}}. 
\end{itemize}

Another widely accepted premise is the manifold assumption, which posits that data often lies approximately on a low-dimensional manifold\rev{; see for example~\cite{belkin2003laplacian,coifman2006diffusion}}. This prior will also be addressed in the present work.


\subsection{The transversality theorem}
\label{sec:transversality} 

Before addressing the stability of the generalized phase retrieval problem, we first need to examine the conditions under which it admits a unique solution. This leads us to the notion of transversality.

\begin{definition}[Transversality] \label{def:transversality}
	Let $V$ a representation of a compact group $H$. We say that a subset $\M \subset V$ is \emph{transverse to the orbits 
		of $H$} if for all $x \in \M$ and $h \in H$, $h \cdot x \in \M$ if and only if $h\cdot x = \pm x$.
\end{definition}
We note that in the literature~\cite{barnett2022geometry}, a pair of manifolds $\M_1, \M_2$ intersects transversely at a point 
$x$ if their tangent spaces are complementary. Because we consider general semi-algebraic sets, our notion of transverse
intersection is purely set-theoretic. 

In~\cite{bendory2024transversality}, a general transversality theorem was derived for semi-algebraic sets in orthogonal or unitary representations of groups. 
It was shown that with a suitable dimension bound, a generic semi-algebraic set is transverse to the orbits of the group action. This, in turn, implies that if a signal lies in a low-dimensional semi-algebraic set, then it can be recovered uniquely from measurements that separate orbits. 
To present the result, we first need to define the  effective dimension of the representation $K$,
defined by the dimension of the representation minus the maximum dimension of the orbits
\begin{equation} \label{eq:K}
	K = \dim V - k(H),    
\end{equation}
where $k(H)$ is the maximum dimension of an $H$ orbit in $V$. 
In general, $k(H)\leq \dim(H)$ and this bound is tight in some cases. 

\begin{theorem}[The transversality theorem~\cite{bendory2024transversality}, informal] \label{thm:transversality}
	If $\M$ is a generic semi-algebraic set of dimension $M$ with $2M<K$,  
	then $\M$ is transverse to the orbits of $H$.
\end{theorem}

\begin{remark}
	A generic semi-algebraic set is obtained by choosing any semi-algebraic set $\M'$ and applying a generic transformation $A$ to it to obtain $\M=A\cdot \M'$. Results in the spirit of \eqref{thm:transversality} were obtained for linear, affine and orthogonal transformations with minor differences. 
\end{remark}

\subsection{Transversality is not enough for stable phase retrieval}
If $V$ is a real representation of a compact group $G$  of the form~\eqref{eq:V}, 
the maximal orbit dimension of the group $H = \prod_{\ell =1}^L \operatorname{O}(N_\ell)$ is
\begin{equation} \label{eq:k_signal_recovery}
	k(H)  = \sum_{\ell =1}^L k_\ell, \quad     k_\ell = \dim \operatorname{O}(N_\ell) - \dim \O(N_\ell - R_\ell).
\end{equation}
\rev{More explicitly, $k_\ell$ is given by 
$$
k_\ell=\twopartdef{R_\ell(N_\ell-R_\ell/2-1/2)}{N_\ell\geq R_\ell,}{\frac{N_\ell^2-N_\ell}{2}}{N_\ell < R_\ell.}$$
}
Then, 
applying Theorem~\ref{thm:transversality} for the action of the group $H = \oplus_{\ell=1}^L \O(N_\ell)$ yields
the following corollary for the second moment.

\begin{corollary}[The injectivity of the second moment~\cite{bendory2024transversality}, informal]\label{cor:secondmoment}
	Let $V$  be a real representation of a compact group $G$ of the form~\eqref{eq:V}. Let $\M \subset V$ be a 
	generic semi-algebraic set of dimension $M$. Then, if $2M<K$, the second moment is injective when restricted to $\M$. Namely, the signal $X$~\eqref{eq:signal} is determined uniquely, up to a sign, from the Gram matrices $X^\top X$~\eqref{eq:Gram}. 
\end{corollary}

\rev{In simple terms, orbit recovery with a generic prior is guaranteed, provided that the dimension $M$ of the prior is less than half of $K$. This dimension $K$ depends on the irreducible decomposition of $V$. To obtain a better intuition of this dependence, we first rewrite \eqref{eq:K} more explicitly as 
\begin{equation}\label{eq:K_elaborate}
K=\sum_{\ell=1}^L \dim \left(V_\ell^{\oplus R_\ell} \right)-k_\ell=
\sum_{\ell\in [L], N_\ell < R_\ell} N_\ell R_\ell-\frac{N_\ell^2-N_\ell}{2}+ \sum_{\ell\in [L], N_\ell \geq  R_\ell} \frac{R_\ell^2+R_\ell}{2}.  
\end{equation}
When considering the phase retrieval problem, in the setting of the discrete Fourier measurements of signals in $V=\mathbb{R}^N $, the decomposition of $V$ into irreducibles (for $N$ even) is a direct sum of two non-isomorphic 1-dimensional irreducibles and $N/2-1$ non-isomorphic 2-dimensional irreducibles. For a derivation,  see~\cite{edidin2023generic}.  Thus, in this setting, we have in total $L=N/2+1 $ non-isomorphic representations, each with a multiplicity of $R_\ell=1$ and dimension $1$ or $2$. In this setting, $K=L=N/2+1$,  and so the corollary gives us injectivity even for priors whose dimensionality is proportional to $N/4$, a quarter of the dimension of the ambient signal space. Similarly, this corollary allows for injectivity for relatively high-dimensional priors in cryo-EM settings, where we consider representations of the 
form $\oplus_{\ell =0}^L (V_\ell)^R$ where $V_\ell$ is the $(2\ell +1)$-dimensional representation of
$\SO(3)$ of spherical harmonic polynomials of degree $\ell$, and $R \geq  2 L+1$. As derived in Section~\ref{sec:cryoEM}, in this model $K \gtrsim 2N/3$ and it follows that we obtain injectivity
for priors whose dimensions are about $N/3$. 
However, for arbitrary representations, this need not be the case. For example, if $V$ consists of a single irreducible representation, then $L=1, R_\ell=1, N_\ell=\dim(V) $ and we have $K=1$, and so our corollary will only apply when the prior is of dimension $M=0$.} 

Corollary~\ref{cor:secondmoment} provides a key insight into the conditions under which the generalized phase retrieval problem becomes solvable. However, its applicability hinges on the exact knowledge of the Gram matrices $X^\top X$. Even a slight inaccuracy in their estimation renders the result invalid. This limitation is critical for real-world scientific and engineering applications (see, for example, Section~\ref{sec:cryoEM}), where noise and model errors are unavoidable.
An illustrative example of this issue is provided by Barnett, Epstein, Greengard, and Magland in \cite{barnett2022geometry}. In their work, they study the phase retrieval problem that arises in coherence diffraction imaging, where the objective is to recover an image with a known support from its power spectrum.
A well-known result, attributed to Hayes \cite{hayes1982reconstruction}, asserts that if the support of the image is sufficiently small, the problem typically admits a unique solution (up to certain symmetries). In the terminology of this paper, this result suggests that the linear subspace of all images with a given support is generically transversed by the power spectrum. 
However, the detailed analysis of~\cite{barnett2022geometry} uncovers an important caveat: In general, the problem is unstable. Specifically, nearly identical power spectra can correspond to entirely different images, even when the supports are the same. This highlights the importance of not only understanding the uniqueness of the problem but also examining how small errors in the observations impact the recovery process. The central focus of our paper is to prove that small errors in the estimation of the Gram matrices lead to correspondingly small errors in signal recovery.

\section{Main results}

\label{sec:Bi-Lipschitz-results}
This section introduces our problem formally and then states the main results and their implications.

\subsection{Problem formulation}
Let $V = \prod_{\ell =1}^L \R^{N_\ell \times R_\ell}$. 
An element of $V$ can be represented by an $L$-tuple of $N_\ell \times
R_\ell$ matrices $X = (X_1, \ldots , X_L)$ or equivalently as a single block matrix of size $\sum_{\ell=1}^L N_\ell \times\sum_{\ell=1}^L R_L$. The second moment is the $L$-tuple of Gram matrices $X^TX:=(X_1^TX_1, \ldots, X_L^TX_L)$, 
which we can identify with a single block diagonal matrix $\diag(X_1^TX_1, \ldots, X_L^TX_L)$. 
Since the matrices $X_\ell^TX_\ell$ are symmetric, positive semidefinite, they have a well-defined symmetric, positive semidefinite square root,
as does the block diagonal matrix $X^TX$. 

We consider the problem of determining priors  $\M\subset V$ for which the 
map $X \to \sqrt{X^TX}$ is bi-Lipschitz when restricted to $\M$.
We will consider priors as in Corollary \ref{cor:secondmoment}, where the second moment determines the signal up to sign. We would like to determine conditions on $\M$ that additionally ensure the second moment is bi-Lipschitz with respect to the 
natural metric on $V/\{\pm 1\}$,
\begin{equation} \label{eq:sigma_metric}
	\rhos(X,Y) = \min\{\norm{X+Y}, \norm{X-Y}\},
\end{equation}
\rev{where throughout $\|\cdot \| $ denotes the Frobenius norm $\|A\|=\sqrt{\sum_{i,j}A_{ij}^2} $.}
\rev{The reason we consider the bi-Lipschitzness of the second moment with respect to 
the metric $d_\sigma$ is because it implies the most natural notion of injectivity on the prior sets $\M$ which we
consider, such as linear subspaces and sparse vectors. As in~\cite{bendory2024transversality}, when we say that the second moment is injective on $\M$ we mean that $x,y \in \M$ have the same second moment if and only if $x = \pm y$; i.e.,  $x$ and $y$ differ by at most a global sign.}

More precisely, our goal is to show that for suitably generic subsets $\M \subset V$, the second moment is bi-Lipschitz in the following sense:

\begin{definition}[bi-Lipschitz] \label{def:biLipschitz}
	Let $\cM \subset V$. We say that the second moment is bi-Lipschitz on $\M$ if 
	the map $$\cM \to \prod_{i=1}^L \R^{R_\ell \times R_\ell};\, X \mapsto \sqrt{X^TX},$$
	satisfies the Lipschitz inequalities
	$$C_1 d_\sigma(X,Y) \leq \norm{\sqrt{X^TX} - \sqrt{Y^TY}} \leq C_2 d_\sigma(X,Y),$$
	for some non-zero, finite constants $C_1,C_2$.
\end{definition}

Notably, we focus on the matrix square root of the second moment instead of the second moment itself,  
because the second moment is quadratic, and thus a simple scaling argument as in \cite[Example 10]{cahill2024towards} 
shows that the second moment is not bi-Lipschitz with respect to the metric $\rhos$.  In contrast, a recent theorem of Derksen~\cite[Theorem 3]{derksen2024bi}  implies that, when applied to the whole matrix domain $V$, 
the square root map
induces a bi-Lipschitz map with respect to the semi-metric 
$$ d_H(X,Y)=\min_{h\in H} \|X-hY\|,$$
where $H = \prod_{\ell =1}^L \O(N_\ell)$. 
Moreover, there are explicit 
Lipschitz bounds of $1$ and $\sqrt{2}$, namely,
\begin{equation} \label{eq.derksenbound}
	d_H(X,Y) \leq \norm{\sqrt{X^TX} - \sqrt{Y^TY}} \leq \sqrt{2} \cdot d_H(X,Y).
\end{equation}
Accordingly, proving  bi-Lipschitzness  of the second moment with respect to $\rhos$, on a subset $\M \subset V$, is equivalent to showing that  there is a bi-Lipschitz equivalence between the metrics $\rhos$ and $\rhoh$ when restricted to $\M$. Moreover, since  
$\rhos(X,Y)\geq \rhoh(X,Y)$ for all $X,Y \in V$, it follows from \eqref{eq.derksenbound} that
$$ \norm{\sqrt{X^TX} - \sqrt{Y^TY}} \leq \sqrt{2} \cdot \rhos(X,Y).$$ 
Thus, we only need to prove that under suitable assumptions on $\M$, there exists some $C>0$ such that,
\begin{equation}\label{eq:nontang}
	C \rhos(X,Y)\leq \rhoh(X,Y), \quad \forall X,Y \in \M.
\end{equation}

\rev{
\begin{remark}
Our focus in this paper is on the metric $\rhos$ obtained from quotienting by the group $\{ \pm 1 \} $, since this is the ambiguity group obtained for generic priors, as explained above. In other settings, it is common to accept a larger ambiguity of all elements of $G$. This would correpsond to a choice of a metric $d_G$ obtained by quotienting over the group $G$ instead of $\{\pm 1\} $. We note that, if $G$ contains $-1$, then   bi-Lipschitzness for $\rhos$ will imply bi-Lipschitzness for this metric as well, as $\rhoh \leq d_G \leq \rhos $. 
\end{remark}
}
\begin{remark}
	In this paper, we choose all matrix norms to be the \rev{Frobenius norm}. This enables us to use the sub-multiplicativity property  $\|AB\|\leq \|A\| \|B\| $ for our proofs later on. This choice does not impact the generality of our result: due to the equivalence of norms on finite-dimensional spaces, the bi-Lipschitz results automatically apply for all matrix norms.
\end{remark}

\subsection{Transversality versus bi-Lipschitzness}
A necessary condition for~\eqref{eq:nontang} to hold is that 
the second moment map is transverse to the orbits of the group $H = \prod_{\ell =1}^L \O(N_\ell)$; see Definition~\ref{def:transversality}.
Otherwise, there would be some $X,Y \in \M$ such that $X\neq Y, -Y $ but $X$ and $Y$ are in the same $H$ orbit. This would imply that $\rhoh(X,Y)=0 $ while $\rhos(X,Y)$ is positive, violating~\eqref{eq:nontang}. 

At first thought, it may seem reasonable to assume that transversality would imply \eqref{eq:nontang}, at least when $\M$ is ``nicely behaved". 
As it turns out, when $\M$ is a linear space, this is indeed true, see Corollary~\ref{thm.linear_bilipschitz}. 
However, in general, the fact that a signal can be uniquely determined from the second moment (up to a global sign) does not preclude
the possibility that this reconstruction is very unstable. Indeed, there are very simple non-linear $\M$ for which this is no longer true; see Example~\ref{ex:line_segment} below.
Our main results, Theorem~\ref{thm:meta} and Theorem~\ref{thm.manifold}, state that a necessary condition for the second moment map to be bi-Lipschitz on a structured semi-algebraic set $\M$
is that a bigger set $\hatM$
is transverse to the orbits of $H = \prod_{\ell =1}^N \O(N_\ell)$. 
For polyhedral sets, 
we give an explicit description of
$\hatM$ in~Theorem~\ref{thm:meta}; when
$\M$ is a compact manifold, then $\hatM$ is the union of the embedded tangent spaces at all points of $\M$, see Theorem~\ref{thm.manifold}.
Combining these results with the transversality theorem, Theorem \ref{thm:transversality} (for full version, see~\cite[Theorem 2.2]{bendory2024transversality}), we obtain a number of corollaries which are stated in Section~\ref{sec:implications} and Section~\ref{sec:manifold}.

\begin{figure}[h]
	\centering
	\includegraphics[width=0.4\linewidth]{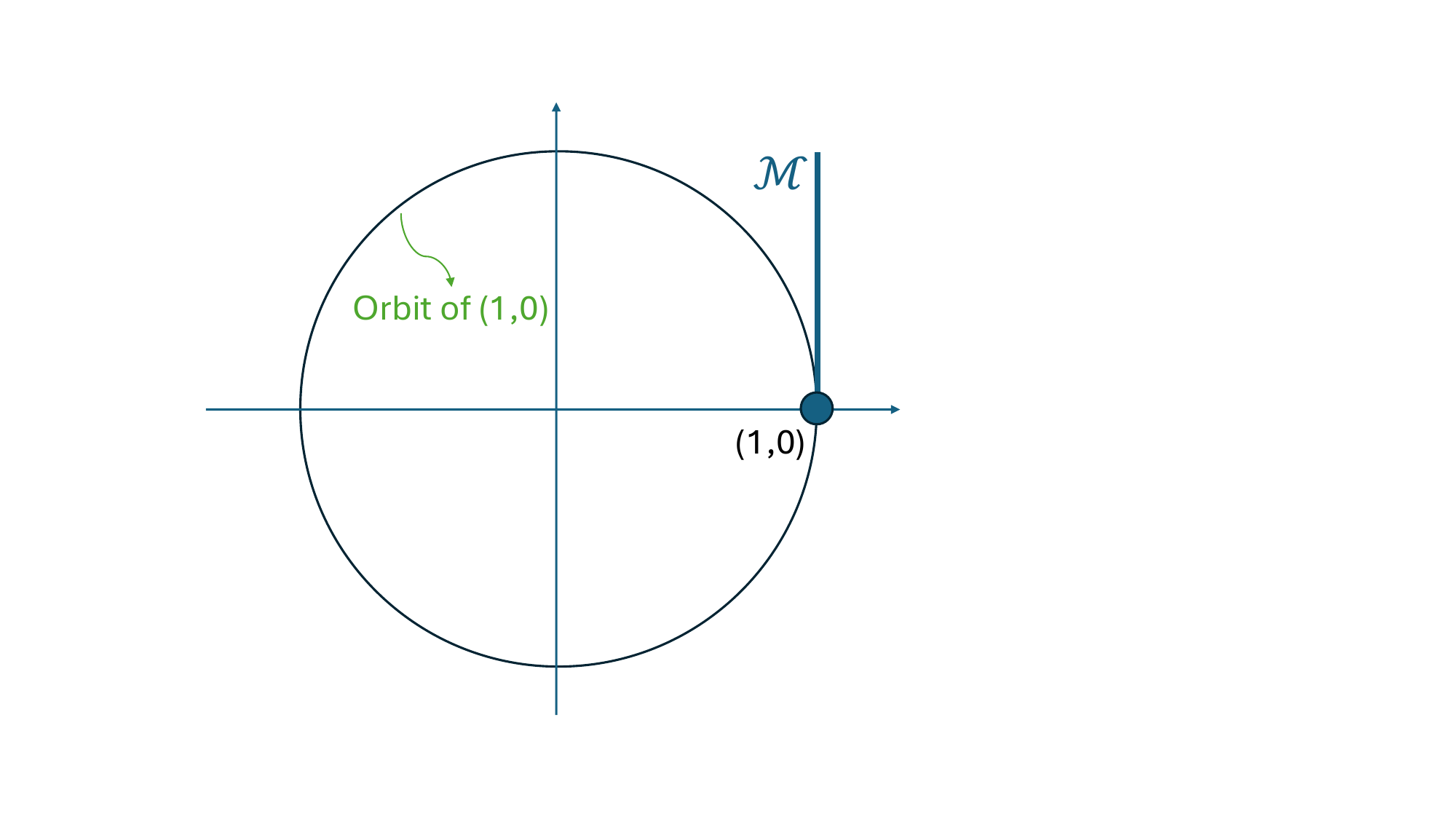}
	\caption{\it \small A set $\M$ which is transverse but does not have the bi-Lipschitz property; see Example~\ref{ex:line_segment}.}
	\label{fig:example}
\end{figure}

We now provide a simple, explicit example of a case where transversality does not imply bi-Lipschitzness.

\begin{example}[Transversality does not imply bi-Lipschitzness] \label{ex:line_segment}
	Consider the case where $V=\R^2$ and $H=\O(2)$. Consider the line segment $\M \subset \R^2 $ defined by 
	$$\M=\{(1,y)\, | \, y\in [0,1] \}.$$  	
	It is obvious that $\M$ is transverse since every two distinct elements in $\M$ have distinct norms. However, the second moment will not be bi-Lipschitz on $\M$. This is because the orbit of the point $(1,0)$, which is just the unit circle, is tangent to $\M$. This situation is visualized in Figure~\ref{fig:example}. Formally, for small values of $y$, the point $(\cos(y),\sin(y))$ will be in the orbit of $(1,0)$, and its distances from the point $(1,y)$ will be of the order $y^2$, namely
	$$\rhoh((1,0),(1,y))\leq \|(\cos(y),\sin(y))-(1,y)\|.  $$
	Using the Taylor approximation of cosine and sine around zero we have  $\cos(y)-1=\frac{y^2}{2}+O(y^4) $ and $\sin(y)-y=-\frac{y^3}{3!}+O(y^5) $, and therefore it is not difficult to bound the expression above by some constant multiplied by $y^2$. On the other hand, for all $y\geq 0$ we will have 
	$$\rhos((1,0),(1,y))=y .$$
	It follows that \eqref{eq:nontang} cannot hold.
\end{example}
We note that, in the example we just discussed,  while the line segment $\M$ is transverse, if we extend the segment to an affine
line it is no longer transverse.
Our main theorem, which we will now state, guarantees in particular that if the affine hull $\hat \M $ of a line segment $\M$ is transverse, then the second moment is bi-Lipschitz on $\M$.

\subsection{Main result}
We now give a full formulation of our main result, Theorem~\ref{thm:meta}. 
In Section~\ref{sec:implications}, we show that it  contains as special cases the three priors discussed in \cite{bendory2024transversality}: (i) linear (ii) sparsity and (iii) generative ReLU networks.
Section~\ref{sec:manifold} discusses the manifold case.  
The following theorem is proved in Section~\ref{sec:proof_meta_thm}.

\begin{theorem}\label{thm:meta}
	Assume that  (i) $\cM$ is either compact or homogeneous\footnote{Recall that a subset $\M \subset V$ is {\em homogeneous} if for all $x \in \M$ and $\lambda \in \R$, $\lambda x \in \M$.} and (ii) $\M$ is contained in a  finite union of affine spaces $\A_1,\ldots,\A_L $. Denote by $\V_i$ the linear subspace $\A_i-\A_i $ and assume that the set 
	$$\hatM:=\cup_{i,j=1}^L \left( \A_i+\V_j \right) $$
	is transverse to the orbits of $H:= \prod_{\ell =1}^L \O(N_\ell)$.
	Then, the second moment is bi-Lipschitz on $\M$. 
\end{theorem}
%

\begin{remark}
We note that the condition that $\M$ is either compact or homogenous is necessary. If we drop this condition, then Theorem~\ref{thm:meta} would imply that affine transverse spaces always have the bi-Lipschitz property. However, the following example shows that this is not true.
\begin{example}
Consider the action of $H=\O(1)^4=\{-1,1\}^4 $ on $\R^4$. For $s,t\in \R$ denote 
	$$v[s,t]:=[s,t,s+1,t+1],$$
	and consider the two dimensional affine space $\M=\{v[s,t]\,| \,s,t\in \R\}$. We claim that $\M$ is transverse but does not have the bi-Lipschitz property. 
	
	The space is transverse, for if $v[s,t]$ and $v[s',t']$ were in the same $H$ orbit, then in particular $|s|=|s'|,|s+1|=|s'+1|$. This implies that $s=s'$, for if for example $s$ were positive while $s'$ were negative, then we would have $|s+1|>|s'+1| $. A symmetric argument shows that $t=t'$. 
	
	To prove that bi-Lipschitzness does not hold, consider for $a>0$ vectors of the form $X_a=v[a,a], Y_a=v[a,-a]$. We will have that 
	$$\rhos(X_a,Y_a)=\|X_a-Y_a\|=\|(0,2a,0,2a)\|=\sqrt{8}a .$$
	On the other hand, by applying the group element $h=[1,-1,1,-1]$ to $Y_a$ we have 
	$$
	\rhoh(X_a,Y_a)\leq \|X_a-hY_a\| =\|[a,a,a+1,a+1]-[a,a,a+1,a-1]\|=2.$$
	Considering the limit where $a\rightarrow \infty$ we see that \eqref{eq:nontang} cannot hold.
\end{example}
\end{remark}

\subsection{Implications} 
\label{sec:implications}

\paragraph{Linear priors.} First, we consider the case of a linear prior, where $\M$ is a linear subspace. In this case, the space $\hatM$ in Theorem~\ref{thm:meta} is equal to $\M$. Thus, the theorem implies that if a linear subspace $\M$ is transverse, then it will always have the bi-Lipschitz property. Moreover, by Proposition~\ref{prop.ATB}, 
a linear subspace is transverse if and only if there does not exist a pair of non-zero matrices $X, Y \in \M$ such that $X^TY$ is skew-symmetric. Combining this with Theorem \ref{thm:transversality} (for full version see  ~\cite[Theorem 2.2]{bendory2024transversality})
we conclude that the second moment is bilipschitz on a 
generic-$\GL(V)$ translate of a linear subspace of dimension $M$ with $2M < K$, where $K$ is defined in~\eqref{eq:K}. Since the set of linear subspaces
of fixed dimension forms a single $\GL(V)$ orbit, we can rephrase this fact as follows.

\begin{corollary}[Linear] \label{thm.linear_bilipschitz} 
	If $\M$ is a linear subspace of $V$, then the following are equivalent: 
	(i) The second moment is injective on $\M$; (ii) the second moment is bi-Lipschitz on $\M$;
	(iii) for every $X, Y \in \M$ both non-zero, the product $X^T Y$ is not skew-symmetric.
	In particular, \rev{for every natural number $M$ satisfying  $2M<K$, (where $K$ is as defined in \eqref{eq:K_elaborate})} the second moment is bi-Lipschitz \rev{when} restricted to a generic linear subspace $\M$ of dimension $M$.
\end{corollary}

\begin{remark} In the case of frame phase retrieval, where $X \in V$ is an $N$-tuple of $1 \times 1$ matrices, then the condition of Corollary~\ref{thm.linear_bilipschitz}~(iii) is equivalent to the statement that $\M$
	does not contain any pair of non-zero vectors with complementary supports. 
\end{remark}

\paragraph{Sparsity.} We next discuss the sparsity prior. If $\mathcal{B}$ is a basis (respectively, an orthonormal basis) for $V$,  then we say that a vector $x \in V$ is $M$-sparse
with respect to the basis~$\mathcal{B}$ if the expansion of $x$ in the (orthonormal) basis $\mathcal{B}$ has at most
$M$ non-zero coefficients. Let us denote by $\M$ the set of all $M$-sparse vectors with respect to a given basis $\mathcal{B}$. 
This set is a finite union of linear spaces of dimension $M$, where each such vector space is defined by the constraint that all but $M$ entries in specified positions are zero. Any two (orthonormal) bases are related by an element of $\GL(V)$ (respectively, $\O(V)$), so
if $\M$ is the set of $M$-sparse vectors with respect to a basis $\mathcal{B}$ and $\M'$ is the set of
$M$-sparse vectors with respect to basis $\B'$, then $\M' = A \cdot \M$ for some $A \in \GL(V)$ (or $\O(V)$ if the bases are orthonormal). 

Applying Theorem \ref{thm:meta} for the case of a finite union of vector spaces $\M=\cup_{i=1}^L \V_i $, we see that the second moment will be bi-Lipschitz on this space, provided that $\hatM=\cup_{i,j=1}^L \V_i+\V_j $ is transverse. Combining this with Theorem \ref{thm:transversality} (for full version see  ~\cite[Theorem 2.2]{bendory2024transversality}) applied to $2M$ dimensional subset $\hatM \subset V$, 
we obtain the following corollary.

\begin{corollary}[Sparsity] \label{cor:sparsity}
	\rev{For every natural number $M$ satisfying $4M<K$ (respectively, $4M+2$), where $K$ is as in \eqref{eq:K_elaborate}}, the second moment is bi-Lipschitz on the set $\M$ of vectors which are $M$-sparse with respect to a generic basis (respectively, a generic orthonormal basis).
\end{corollary}

\paragraph{Deep ReLu neural networks.}Finally, we consider priors generated by ReLU networks. Namely, we consider signals $x\in V\equiv \mathbb{R}^N$ which are of the form~\eqref{eq:relu}, where $z$ resides in a low $M$-dimensional latent space $[0,1]^M $. We make no requirements on the dimensions of $A_1,\ldots,A_\ell$ except that the composition is well defined, and  $A_\ell$ is \rev{ a square  matrix}.  
Under these conditions, the following corollary shows that almost all (in the Lebesgue sense) choices of neural network parameters lead to a signal that can be stably inverted.

\begin{corollary}[ReLU networks] \label{cor:ReLu}
	For arbitrary affine transformations $A_1, \ldots A_{\ell-1}$ and a generic choice of the parameters of $A_\ell\in \mathbb{R}^{N\times N}$, if $4M<K$ \rev{(for $K$ as in \eqref{eq:K_elaborate})}, then the second moment is bi-Lipschitz on the image of the neural network defined in~\eqref{eq:relu}.
\end{corollary}
\begin{proof}
	For any fixed choice of $A_1,\ldots,A_{\ell-1}$, the image of \rev{$[0,1]^M $} under the piecewise linear function $\eta \circ A_{\ell-1}\circ \ldots \circ \eta \circ A_1 $ is a finite union of compact polytopes whose dimension is at most $M$. Taking this set as $\M$, it is contained in a finite union of affine spaces of dimension at most $M$, and hence the second moment is bi-Lipschitz on $\M$, provided that the set $\hatM$, which has dimension at most $2M$, is transverse. By Theorem \ref{thm:transversality} (for full version see  ~\cite[Theorem 2.2]{bendory2024transversality}), a generic affine translate of $\hatM$ will be transverse, if $4M<K $. Since a generic affine translate  $A_\ell \cdot \M$ is the image of $A_\ell  \circ \eta \circ A_{\ell-1} \circ \ldots\circ \eta \circ A_1$, \rev{we deduce that  for a generic choice of a final affine map $A_\ell$,  the second moment is bi-Lipschitz on the set of signals spanned by the ReLU network.}
\end{proof}

\subsection{The manifold case} \label{sec:manifold}
All the priors we have considered up to now are `piecewise linear', meaning that they are representable as a finite union of linear spaces or compact polytopes. Next, we state analogous results for the case where $\M$ is a manifold. We begin with two definitions.

\begin{definition}
	A homogeneous closed subset $\cM \subset V$ is called a {\em homogeneous manifold} if it is smooth at every \rev{nonzero $X \in \cM$.}
    \end{definition}

\begin{definition}
	A compact submanifold $\cM \subset V$ is \rev{said to be} {\em well-situated} if $0 \notin \cM$ and either $\cM = -\cM$ or $\cM \cap -\cM =\emptyset$.
\end{definition}

The {\em well-situated} condition we impose is necessitated by the proof of Theorem~\ref{thm.manifold}. Note that if $\M$
is well-situated, then $A \cdot \M$ is also well-situated for any $A \in \GL(V)$. Also, for any compact submanifold $\M \subset V$, 
there is a translation $T \in \Aff(V)$ such that $T \cdot \M$ is well-situated.  An example of a well-situated submanifold
is the unit sphere centered at the origin.

\begin{theorem}[Manifold] \label{thm.manifold}
	Assume that  $\cM \subset V$ is either a homogeneous submanifold or a compact well-situated manifold which is transverse to the orbits
	of $H = \prod_{\ell =1}^L \O(N_\ell)$. Then, the second moment is bi-Lipschitz on $\cM$ if for all $0\neq  X  \in \cM$ the embedded tangent space $T_X\cM$ (viewed as an affine subspace of $V$) is transverse to the orbits of $H$.
\end{theorem}

The proof of Theorem~\ref{thm.manifold} is an adaptation of the proof of Theorem~\ref{thm:meta} and is given in Section~\ref{appendix:manifold_proof}.
Once again, we can apply Theorem \ref{thm:transversality} (for full version see  ~\cite[Theorem 2.2]{bendory2024transversality}) to the union of the embedded tangent spaces $T_X\M$ over all $X \in \M$
and obtain the following corollary.

\begin{corollary} \label{cor:manifold_dimension_bounds} If $\cM$ is a semi-algebraic homogeneous submanifold or a 
	well-situated compact semi-algebraic submanifold
	of dimension $M$ with $4M<K$ \rev{(for $K$ as in \eqref{eq:K_elaborate})}, then for a generic linear transformation $A \in \GL(V)$, the second moment is bi-Lipschitz on $A \cdot \cM$.
\end{corollary}

\begin{remark} When $4M<K$, the semi-algebraic condition in Corollary~\ref{cor:manifold_dimension_bounds} is not strong because a celebrated theorem of John Nash~\cite{nash1952real} implies that any compact embedded submanifold can be arbitrarily closely approximated by algebraic submanifolds. Stronger results have been proved in~\cite{akbulut1992approximating}.
\end{remark}

\section{Proofs of Theorem~\ref{thm:meta} and Theorem~\ref{thm.manifold}}
\label{sec:proof_meta_thm}

Let $H = \prod_{\ell =1}^L \O(N_\ell)$ and let $\rhog$ be the following surrogate metric to $\rhoh$, which compares matrices based on the distances of their Gram matrices:
$$ \rhog(X,Y) := \sqrt{\norm{X^T X-Y^T Y}}.$$
Since $X,Y$ lie in the same $H$ orbit if and only if their Gram matrices are equal, we have
that  $\rhoh(X,Y)=0$ if and only if $\rhog(X,Y)=0$. 
The metric $\rhog$ is convenient to work with for our proof, essentially since $d^2_{\textup{Gram}}$
is defined by a simple quadratic form. Moreover, although $\rhoh$ is not lower-Lipschitz with respect to $\rhog$, 
it does hold that $\rhoh$ is \emph{locally lower-Lipschitz} with respect to $d^2_{\textup{Gram}}$, as the following lemma shows.

\begin{lemma}\label{thm_rhog_local_lowlip} 
	Let $B_r\of{X_0}$ be the Frobenius norm ball of radius $r$ centered at $X_0$. Then, for each $X_0 \in V$ and $r>0$,  there exists a constant $C>0$ such that for any $X,Y \in B_r\of{X_0}$,
	\begin{equation}\label{eq_rhog_local_lowlip}
		d_{\textup{Gram}}(X,Y)^2 \leq C \rhoh\of{X,Y}.
	\end{equation}
\end{lemma}

\begin{proof}
Note that $\rhog$ can be written as
\begin{equation}\label{eq_rhog_alternative}
	\rhog\of{X,Y}^2 
	\eqx \tfrac{1}{2}\norm{ \br{X-Y}^T\br{X+Y} + \br{X+Y}^T\br{X-Y} }.
\end{equation}
Thus,
\begin{equation*}
	\begin{split}
		\rhog\of{X,Y}^2 
		\leqx &\tfrac{1}{2}\norm{ \br{X-Y}^T\br{X+Y} } + \tfrac{1}{2}\norm{\br{X+Y}^T\br{X-Y} }
		\\
		\leqx &\tfrac{1}{2}\norm{ \br{X-Y}^T} \cdot \norm{X+Y} + \tfrac{1}{2}\norm{\br{X+Y}^T} \cdot \norm{X-Y}
		\\
		\eqx &\norm{X+Y} \cdot \norm{X-Y}
		\\
		\eqx &\norm{2 X_0 + \br{X-X_0} + \br{Y-X_0}} \cdot \norm{X-Y}
		\\
		\leqx &\br{\norm{2 X_0} + \norm{X-X_0} + \norm{Y-X_0}} \cdot \norm{X-Y}
		\\
		\leqx &\br{2 \norm{X_0} + r + r} \cdot \norm{X-Y}.
	\end{split}    
\end{equation*}
Since $\rhog\of{X,Y}=\rhog\of{X,Q Y}$ for all $Q \in H = \prod_{\ell =1}^L \O(N_\ell)$, by replacing $Y$ above with $Q Y$ we get
\begin{equation}
	\rhog\of{X,Y}^2
	\leq
	2\br{\norm{X_0}+r} \cdot \norm{X-Q Y},
	\quad
	\forall Q \in H
\end{equation}
Thus,~\eqref{eq_rhog_local_lowlip} holds with $C=2\br{\norm{X_0}+r}$.
\end{proof}
Our proof \rev{of Theorem~\ref{thm:meta}} is based on considering the first-order expansion of the quadratic function $X^TX $ around a matrix $A$. Note that if $X_1=A+B_1, X_2=A+B_2$ (where we think of $B_1,B_2$ as a small perturbation of magnitude $\epsilon$) then 
$$X_1^TX_1-X_2^TX_2=(B_1-B_2)^TA+A^T(B_1-B_2)+O(\epsilon^2).$$
The term above will thus vanish (up to an $O(\epsilon^2)$ correction term) if the matrix $(B_1-B_2)^TA $ is anti-symmetric. As we will see, the key to ensuring bi-Lipschitzness will be assure that this matrix will not be anti-symmetric. 
The following proposition connects this condition \rev{with} the notion of transversality.

\begin{proposition} \label{prop.ATB}
	Two non-zero matrices $A, B \in V$ satisfy the condition that $A^TB$ is skew symmetric if and only if there is an element
	$R \in H = \prod_{\ell =1}^L \O(N_\ell)$ such that $A-B = R(A +B)$.
\end{proposition}
\begin{proof}[Proof of Proposition~\ref{prop.ATB}]
Let us denote $X = A+B$ and $Y = A-B$, and assume that  $Y = RX$ for some $R\in H$.  Then $A = \rev{\frac{1}{2}}(X + RX)$ and $B = \rev{\frac{1}{2}}(X - RX)$. 
Thus, 
\begin{align*}A^TB &= \rev{\frac{1}{4}}(X +RX)^T(X -RX)\\
	&= \rev{\frac{1}{4}}(X^TX -X^TR^TRX + X^TR^TX - X^TRX) =\rev{\frac{1}{4}}( X^TR^TX - X^TRX).\end{align*}
Likewise $B^TA = \rev{\frac{1}{4}}(-X^TR^TX + X^TRX)$ so $A^TB$ is skew-symmetric.

Conversely, suppose that $A^TB$ is skew-symmetric.
Let $X = \rev{\frac{1}{2}}(A+B)$ and $Y= \rev{\frac{1}{2}}(A-B)$. Then the Gram matrix 
$$X^TX = \rev{\frac{1}{4}}(A^T+B^T)(A+B)=\rev{\frac{1}{4}}(A^TA + B^TB)$$ because $A^TB + B^TA =0$. Likewise, 
$$Y^TY = \rev{\frac{1}{4}}(A^T - B^T)(A^T -B^T) = \rev{\frac{1}{4}}(A^TA + B^TB).$$
Since $X$ and $Y$ have the same Gram matrices we must have $Y = RX$ for some $R \in H= \prod_{\ell=1}^L \O(N_\ell)$.
\end{proof}
\begin{proof}[Proof of Theorem~\ref{thm:meta}] We assume by contradiction that all the conditions of the theorem hold, but the second moment is not bi-Lipschitz. This implies that there exists sequences $X_k,Y_k \in \M$ such that 
\begin{equation}\label{eq_violation_rhoo}
	\forall k  \quad  	\rhos\of{X_k,Y_k} > 0,\text{ and } 		\lim_{k\to\infty} \frac{\rhoh\of{X_k,Y_k}}{\rhos\of{X_k,Y_k}} = 0.
\end{equation}

The following technical lemma will be proved below.
\begin{lemma}\label{thm_global_to_local}
	Under our assumptions on $\M$, suppose there exist two sequences $X_k, Y_k \in \M$ that satisfy \eqref{eq_violation_rhoo}.
	Then, these sequences can be chosen such that $X_k \rightarrow X_0, Y_k \rightarrow sX_0 $,
	for some $X_0 \in \cM$, $X_0 \neq 0$ and $s\in \{-1,1\}$. Moreover, all of the $X_k$ can be chosen to be in a single affine space $\A_i $, the $Y_k$ in an affine space $\A_j$, and $X_0\in \A_i \cap \A_j$. 
\end{lemma}

Now, consider the expression 
\begin{equation} \label{eq.skew_deviation}
	c=\min_{S\in \V_i+\V_j, \|S\|= 1} \|(X_0^TS)+(S^T X_0)\|.
\end{equation}
This minimum is obtained since the domain we minimize over is compact, and is clearly non-negative. To show that it is strictly positive, we will assume that $X_0^TS$ is anti-symmetric and show that this implies $S=0$.  Recall that by Proposition \ref{prop.ATB}, the matrix $X_0^TS $ is anti-symmetric if and only if $X_0+S $ and $X_0-S$ are related by a rotation in $H$. We can write $S=A-B$ where $A\in \V_i, B\in \V_j$, and we then obtain that $(A+X_0)-B $ and $(X_0+B)-A $ are related by a rotation in $H$. \rev{Note that
$(A+X_0) - B \in \A_i - \V_j \subset \hatM$ and $(B + X_0) - A \in A_j - V_i \subset \hatM$.
By hypothesis, the set $\hatM$ is transverse so}
$$A+X_0-B=X_0+B-A \text{ or } A+X_0-B=-(X_0+B-A). $$
\rev{This} is equivalent to $B=A$ or $X_0=0 $. But $X_0\neq 0$ and thus we deduce that
$S=A-B=0$. We thus have that  the minimum $c$ is strictly positive. 

Next, we denote $X_k = X_0 + \Delta X_k$ and $Y_k =sX_0 + \Delta Y_k$. We know that $\Delta X_k$ and  $\Delta Y_k $ converge to zero.  By considering the difference between their second moments, we obtain
\begin{equation}\label{eq_rhog_expansion}
	\begin{split}
		d^2_{\textup{Gram}}(X_k,Y_k)
		\eqx & d^2_{\textup{Gram}}(X_k,sY_k)\\
		\eqx & \tfrac{1}{2}\norm{ \br{X_k-sY_k}^T\br{X_k+sY_k} + \br{X_k+sY_k}^T\br{X_k-sY_k} } 
		\\
        \eqx &\tfrac{1}{2}\| \br{\Delta X_k -s \Delta Y_k}^T\br{2X_0 + \Delta X_k + s\Delta Y_k} \\
        & \quad \quad  + \br{2X_0 + \Delta X_k+ s\Delta Y_k}^T\br{\Delta X_k- s\Delta Y_k} \| \\ 
        \geqx &\tfrac{1}{2}\norm{ \br{\Delta X_k - s\Delta Y_k}^T\br{2X_0} + \br{2X_0}^T\br{\Delta X_k- s\Delta Y_k} }
		\\ - &\tfrac{1}{2}\norm{\br{\Delta X_k - s\Delta Y_k}^T\br{\Delta X_k + s\Delta Y_k} + \br{\Delta X_k +s \Delta Y_k}^T\br{\Delta X_k - s\Delta Y_k}}
		\\ \geqx & c \cdot \|\Delta X_k- s\Delta Y_k\|-\|\Delta X_k- s\Delta Y_k\| \cdot \|\Delta X_k+ s\Delta Y_k\| \;\;   (*)
		\\ \eqx & \norm{\Delta X_k - s\Delta Y_k} \cdot \br{ \frac{c}{2} - \norm{\Delta X_k + s\Delta Y_k} }
		\\ \eqx & \norm{X_k - sY_k} \cdot \br{ \frac{c}{2} - \norm{\Delta X_k + s\Delta Y_k} }
		\\ \geqx & \norm{X_k -s Y_k} \cdot \br{ \frac{c}{2} - \norm{\Delta X_k} - \norm{\Delta Y_k} },
	\end{split}
\end{equation}
where (*) follows from~\eqref{eq.skew_deviation}.  
It follows that 	\begin{equation}\label{eq_anti_symmetric_limit}
	\begin{split}
		C\lim_{k \to \infty} \frac {\rhoh\of{X_k,Y_k}} {\rhos\of{X_k,Y_k}} &\stackrel{\eqref{eq_rhog_local_lowlip}}{\geq} 
		\lim_{k \to \infty} \frac {\rhog\of{X_k,Y_k}^2} {\rhos\of{X_k,Y_k}}\\
		&\geqx
		\lim_{k \to \infty}
		\frac {\rhog\of{X_k,Y_k}^2} {\norm{X_k-sY_k}}
		\geqx
		\lim_{k \to \infty}
		\frac{c}{2} - \norm{\Delta X_k} - \norm{\Delta Y_k}
		= \frac{c}{2} > 0,
	\end{split}
\end{equation}
which is a contradiction. This concludes the proof of the theorem. 
\end{proof}%
\begin{proof}[Proof of Lemma~\ref{thm_global_to_local}]
Without loss of generality, we assume that $\M$ is defined by intersection of  \rev{a ball} of radius at least one with the union \rev{of the affine} spaces. This is because,  if $\M$ is compact, then it is contained in such a set, and if $\M$ is homogeneous, then since both $\rhoh$ and $\rhos$ are also homogeneous, it is sufficient to prove bi-Lipschitzness when restricting to a unit ball. 

Let $X_k, Y_k$ that satisfy~\eqref{eq_violation_rhoo}.
By moving to a subsequence, we can assume that all  $X_k$ are in the same affine space $\A_i$, and all $Y_k$ are in the same affine space $\A_j $.  We next rule out the option that both $X_k$ and $Y_k$ converge to zero: if they do, then the limit zero is in both $\A_i$ and $\A_j$, and therefore they are both vector spaces. By defining $M_k=\max\{\|X_k\|,\|Y_k\| \} $ and  replacing $X_k$ and $Y_k$ with $\frac{1}{M_k}X_k  $ and $\frac{1}{M_k}Y_k $, we will obtain new sequences which are in the vector spaces $\A_i$ and $\A_j$, respectively, have \rev{norm at} most one (and hence are in $\M$), and satisfy \eqref{eq_violation_rhoo}, since $\rhoh$ and $\rhos$ are homogeneous.
Clearly, these new normalized sequences cannot both converge to zero, since for each $k$ either $X_k$ or $Y_k$ has unit norm.

Next, since $\M$ is compact, \rev{$\rhos$} is bounded, and thus $\rhoh(X_k,Y_k)\rightarrow 0 $
\rev{since we assume that~\eqref{eq_violation_rhoo} holds.} It follows that also $\rhos(X_k,Y_k)\rightarrow 0$ since $\M \subset \rev{\widehat{\M}}$ is tranvserse by hypothesis.
Next, invoking compactness again, we can assume by moving to a subsequence that $\X_k $ and $\Y_k$ converge to some limit $\X_0$ and $\Y_0$ respectively. We then have $\rhos(X_0,Y_0)=0 $ so that $X_0=sY_0$ for \rev{some} $s\in \{-1,1\}$.  Since we saw that both sequences do not converge simultaneously to zero, we deduce that $X_0\neq 0$. This concludes the proof of the lemma. 
\end{proof}
\subsection{Proof of Theorem~\ref{thm.manifold}}\label{appendix:manifold_proof}
We follow the strategy of the proof of Theorem~\ref{thm:meta}.
Again, if the second moment map is not bi-Lipschitz on $\M$, then there are sequences $X_k,Y_k \in \M$ such that 
\begin{equation}
	\forall k  \quad  	\rhos\of{X_k,Y_k} > 0,\text{ and } 		\lim_{k\to\infty} \frac{\rhoh\of{X_k,Y_k}}{\rhos\of{X_k,Y_k}} = 0.
\end{equation}
First, we show that our hypothesis on $\M$ ensures that the conclusion of Lemma~\ref{thm_global_to_local} holds for $\cM$. Namely, we
show that there exists $X_0 \neq 0$ and sequences $X_k \to X_0$, $Y_k \to  X_0 \neq 0$ such that
\begin{equation} \lim_{k \rev{\to} \infty} \frac{d_H(X_k,Y_k)}{d_\sigma(X_k,Y_k)} =0.\end{equation}\label{eq.limit}
Next, replace $X_k, Y_k$ with convergent subsequences. If $\M$ is homogeneous we can, after replacing $X_k, Y_k$
by $\frac{X_k}{M_k},\frac{Y_k}{M_k}$ where $M_k = \max\{ \norm{X_k}, \norm{Y_k}\}$,  
assume that the sequences are bounded by 1.
Finally replacing $Y_k$ with $-Y_k$ as necessary we can 
assume $d_\sigma(X_k,Y_k) = \norm{X_k - Y_k}$. Since $d_\sigma(X_k,Y_k)$ is bounded it follows that $d_H(X_k, Y_k) \to 0$. Since
$\M$ is assumed to be transverse to the orbits of $H$ we must also have that $d_\sigma(X_k,Y_k) \to 0$ and we can take $X_0 = \lim_{k \to \infty}X_k \neq 0$ because either $X_k$ or $Y_k$ has norm one for each $k$.

If $\cM$ is a compact manifold which is transverse to the orbits of $H$ we can use compactness to bound the magnitudes $\norm{X_k}$ and $\norm{Y_k}$ and again find subsequences such that $X_k \to X_0$ and $Y_k \to \pm X_0$. By hypothesis, $0 \notin \M$ so
$X_0 \neq 0$. Moreover, if $\M$ is well-situated then we \rev{can} ensure that $X_k$ and $Y_k$ both converge to $X_0$

We now adapt the second part of the \rev{proof} of Theorem~\ref{thm:meta}. 
Let $\V$ be the linear subspace $T_{X_0}\M - T_{X_0}\M$ which is by definition, the union of all tangent vectors at $X_0$. 
Let $c = \min_{S \in \V,\norm{S} =1} \norm{X_0^TS + S^TX_0}$.
As in the proof of Theorem~\ref{thm:meta}, the fact that the tangent space $T_{X_0}\M$ is transverse to the orbits of $H$ implies
that $c> 0$.

\rev{Temporarily surpressing the index $k$, let $X= X_0 + \Delta X$, $Y = X_0 + \Delta Y$. Now write $\Delta X = T\Delta X + \Delta X^{(2)}$,  where
$T\Delta X$ is the projection of $\Delta X$ onto the linear space  $\V$ spanned by the tangent vectors
at $X_0$, so $X_0 + T \Delta X \in T_{X_0}\M$.}
Then, as before, we can write, 
\begin{equation}\label{eq_rhog_expansion_manifold}
	\rev{\begin{split}
		\rhog\of{X,Y}^2
		\eqx &\tfrac{1}{2}\norm{ \br{X-Y}^T\br{X+Y} + \br{X+Y}^T\br{X-Y} } 
		\\ \eqx &\tfrac{1}{2}\norm{ \br{\Delta X - \Delta Y}^T\br{2X_0 + \Delta X + \Delta Y} + \br{2X_0 + \Delta X+ \Delta Y}^T\br{\Delta X- \Delta Y} } 
		\\ \geqx &\tfrac{1}{2}\norm{ \br{\Delta X - \Delta Y}^T\br{2X_0} + \br{2X_0}^T\br{\Delta X- \Delta Y} }
		\\ - &\tfrac{1}{2}\norm{\br{\Delta X - \Delta Y}^T\br{\Delta X + \Delta Y} + \br{\Delta X + \Delta Y}^T\br{\Delta X - \Delta Y}}
		\\ \geqx &\tfrac{1}{2}\norm{ \br{T\Delta X - T\Delta Y}^T\br{2X_0} + \br{2X_0}^T\br{T\Delta X- T\Delta Y} }
        \\ - & \tfrac{1}{2} \norm{\br{\Delta X^{(2)} - \Delta Y^{(2)}}^T\br{2X_0} + \br{2X_0}^T
        \br{\Delta X^{(2)} - \Delta Y^{(2)}}}
        \\ - &\tfrac{1}{2}\norm{\br{\Delta X - \Delta Y}^T\br{\Delta X + \Delta Y} + \br{\Delta X + \Delta Y}^T\br{\Delta X - \Delta Y}}
		\\ \geqx & \frac{c}{2}\br{ \|T\Delta X- T\Delta Y\| -\| \Delta X^{(2)} - \Delta Y^{(2)} \|} -\|\Delta X- \Delta Y\| \cdot \|\Delta X+ \Delta Y\|.
		%
		%
		%
	\end{split}}
\end{equation}
It follows that 
\begin{equation}
	\label{eq_anti_symmetric_limit2}
	\begin{split}
		0 & = \lim_{k \to \infty} \frac {\rhog\of{X_k,Y_k}^2} {\rhos\of{X_k,Y_k}}\\
		& \geqx 
		\lim_{k \to \infty}
		\frac {\rhog\of{X_k,Y_k}^2} {\norm{X_k-Y_k}}
		\\ &\geqx
		\lim_{k \to \infty}
		\frac{c}{2}\frac{\norm{T\Delta X_k - T \Delta Y_k} - \rev{\norm{\Delta X_k^{(2)} - \Delta X_k^{(2)}}}}{\norm {X_k - Y_k}}- \norm{\Delta X_k} - \norm{\Delta Y_k}
		\\ & = \lim_{k \to \infty}
		\frac{c}{2}\frac{\norm{T\Delta X_k - T \Delta Y_k}}{\norm{X_k - Y_k}} - \lim_{k \to \infty} \rev{\frac{c}{2}\frac{\norm{\Delta X_k^{(2)} -  \Delta Y_k^{(2)}}}{\norm{X_k - Y_k}}}.
	\end{split}
\end{equation}

\begin{lemma} \label{lem:tanprojection} $\lim_{k \to \infty} \frac{\norm{T\Delta X_k - T \Delta Y_k}}{\norm {X_k - Y_k}} =1$
and $\lim_{k \to \infty} \frac{\norm{\Delta X_k^{(2)} -  \Delta Y_k}^{(2)}}{\norm {X_k - Y_k}} =0$
\end{lemma}

Given the \rev{lemma,} and we obtain a contradiction to the statement that the second moment map is not bi-Lipschitz.

\begin{proof}[Proof of Lemma \ref{lem:tanprojection}]
	To prove the \rev{lemma} we use the fact~\cite[Theorem 5.8]{lee2013introduction} that in a neighborhood of $X_0 \in V$ there are local coordinates
	$(z_1, \ldots , z_N)$ such that in these coordinates $\M$ is locally parametrized as $(z_1, \ldots , z_M)$. This means
	that there is a differentiable function $\varphi \colon (U, 0) \subset \R^M \to \R^N$ such that $\varphi(U)$ is an open neighborhood
	of $X_0 \in \M$ and $\varphi$ is a diffeomorphism onto its image. Thus, the tangent space to $\M$ can be parametrized as the affine space
	$X_0 + D\varphi(0)(\R^M)$. On the other hand, for $\rev{k \gg 0}$ both $X_k, Y_k$ lie in the image of $\varphi(U)$ so we may assume
	that there are sequences $t_k \to 0$ and $s_k \to 0$ in $U$ with $\varphi(t_k) = X_k$ and $\varphi(s_k) = Y_k$. Using a first-order
	Taylor approximation,
	we have $X_k = X_0 + D\varphi(0)(t_k) + Z_k$, where $\norm{Z_k}= O(\norm{t_k}^2)$ and an analogous statement for the $Y_k$. But 
    \rev{$D\varphi(0)(t_k)$
	is exactly $T\Delta X_k$ and therefore $Z_k = \Delta X_k^{(2)}$},
    so the \rev{lemma} follows.
\end{proof}

\section{Single-particle cryo-electron microscopy and sample complexity analysis} 
\label{sec:cryoEM}

Over the past decade, cryo-electron microscopy (cryo-EM) has become indispensable in structural biology, driven by technological and algorithmic advancements that have significantly improved resolution. This technique enables the reconstruction of biomolecules in their native state and has successfully resolved atomic structures of proteins that were previously thought impossible, e.g.,~\cite{bartesaghi20152,lyumkis2019challenges}. The number of protein structures determined by cryo-EM is increasing exponentially and is projected to surpass alternative technologies within the next few years.

\subsection{Mathematical model}
In a cryo-EM experiment, an electron microscope captures a large image that contains multiple 2D tomographic projections of the target molecules. The 3D orientations of these projections are unknown, and the low electron doses used in microscopy result in extremely low signal-to-noise ratios (SNRs).  
Computationally, the primary challenge is to efficiently estimate a 3D structure in this low-SNR regime while the 3-D orientations of the projections are unknown. Under certain simplifications, cryo-EM observations (projections) can be modeled as~\cite{bendory2020single, singer2020computational}:  
\begin{eqnarray} \label{eq:cryoEM}  
	y = P(g \cdot x) + \varepsilon,  
\end{eqnarray}  
where \( P \) represents a tomographic projection, \( G \) is the group of 3-D rotations \( \SO(3) \), and \( x:\R^3 \to \R \) denotes the electrostatic potential of the molecule to be recovered.  
A convenient way to represent the 3D molecular structure  is:
\begin{equation} \label{eq:molecule}
	x(r, \theta, \phi) = \sum_{\ell=0}^{L} \sum_{m=-\ell}^{\ell} X_{\ell, m}(r) Y_{\ell}^m(\theta, \phi),  
\end{equation}
where \( \rev{r\in\{1,\dots,R\}}, \) is the radial frequency, \( Y_{\ell}^m(\theta, \phi) \) are spherical harmonics, and \( L \) denotes the bandlimit.
This means that recovering the molecular structure is equivalent to recovering the \( L \)-tuple of matrices \( X=(X_1, \dots, X_L) \), where each \( X_\ell \) is an \( N_\ell \times R_\ell \) matrix containing the spherical harmonic coefficients \( \{X_{\ell, m}(r)\} \).
This model is widely adopted in the cryo-EM literature, e.g.,~\cite{bandeira2020non,bendory2023autocorrelation}. 
Remarkably, it was shown that 
the second moment of~\eqref{eq:cryoEM} is invariant to the tomographic projection~\cite{kam1980reconstruction, bendory2024sample} and thus it can be understood as a special case of the MRA model~\eqref{eq:mra} with $G=\SO(3)$. 
Thus,  signal recovery from the second moment is equivalent to factorizing a tuple of Gram matrices. 
Namely, it is a special case of the generalized phase retrieval problem with respect to the special orthogonal group $\SO(3)$.
\rev{A closely related application is the sub-tomogram averaging problem, an essential component of the emerging technology of cryo-electron tomography~\cite{watson2024advances}, which parallels the cryo-EM model~\eqref{eq:cryoEM} apart from the tomographic projection.}

\subsection{Second moment analysis and transversality}
The determination of the 3-D molecular structure from the second moment in cryo-EM has received considerable attention for several reasons. 
First, the moments can be reliably estimated from the data by averaging over the empirical moments, without the need to estimate the unknown 3-D orientations as an intermediate step (the rotations are treated as nuisance variables). This is crucial since estimating the orientations becomes inherently difficult as the noise level increases. 
Moreover,  it has been shown that the highest-order moment necessary to recover a signal determines the sample complexity of cryo-EM in high-noise regimes~\cite{abbe2018estimation} (this is also true for MRA models~\eqref{eq:mra}). Since recovering from the first moment (the average) is impossible, recovering from the second moment leads to an optimal estimation rate. 
In particular, if recovery from the second moment is possible, then the asymptotic sample complexity of cryo-EM and MRA models is $n/\sigma^4\to\infty$, where $\sigma^2$ is the variance of the noise.
This naturally raises the question: under what conditions can we recover a molecular structure from its second moment? More specifically, what additional information about the molecular structure can resolve the ambiguity in the second moment arising from the unknown orthogonal matrices?
This question was addressed in~\cite{bendory2024transversality} using the transversality theorem (Theorem~\ref{thm:transversality}). For the cryo-EM model~\eqref{eq:cryoEM} with the molecular representation~\eqref{eq:molecule}, it was shown that if \( R \geq 2L + 1 \), then
\begin{equation} \label{eq:Kcryo}
	K =(L+1)\left(R (L+1) - {\frac{L(4L+5)}{6}}\right)\approx L^2\left(R\rev{-}\frac{2L}{3}\right).    
\end{equation}

\subsection{Implications of the bi-Lipschitz results to cryo-EM}

This work extends the results of~\cite{bendory2024transversality} by showing that under similar conditions, the map between the molecular structure and its second moment is stable. 
This stability is critical for second-moment-based algorithms---such as those in~\cite{donatelli2015iterative,bhamre2015orthogonal,levin20183d,bendory2023autocorrelation}---which achieve the optimal estimation rate.
As a result, molecular structure recovery becomes feasible with fewer observations, an essential advantage for applications like heterogeneity analysis~\cite{toader2023methods} and scenarios where data acquisition is a major bottleneck.

\rev{Importantly, these low-dimensional structures were identified early in the cryo-EM literature and have since been leveraged for a variety of tasks. For instance, PCA-based techniques, which assume that the data lie on a low-dimensional subspace, have proven highly effective for image denoising~\cite{van1981use,bhamre2016denoising,zhao2016fast} and 3-D representation~\cite{fraiman2025so}. Other approaches exploit sparsity to design computationally efficient algorithms~\cite{bendory2023autocorrelation,vonesch2011fast,chen2021deep}. More recently, a range of methods implicitly leveraging the low-dimensional structure of neural networks---more complex in practice than the simple model in~\eqref{eq:relu}---have been applied across different stages of the cryo-EM computational pipeline.} 

Using the estimate for $K$ given in \eqref{eq:Kcryo} and applying Corollaries~\ref{thm.linear_bilipschitz}-\ref{cor:ReLu}, we obtain the following corollary for the cryo-EM model.

\begin{corollary}[cryo-EM] \label{cor:cryo}
	Consider the cryo-EM model~\eqref{eq:cryoEM} and the molecule model~\eqref{eq:molecule} with $R\geq 2L+1$. Then, the second moment is bi-Lipschitz if the molecular structure is restricted to $\M$, where 
	\begin{itemize}
		\item $\M$ is a generic linear subspace of dimension $M\lesssim \frac{L^2}{2}(R-\frac{2L}{3})$;
		\item $\M$ is the set of vectors which are $M$-sparse with respect to a generic basis
		of dimension $M\lesssim \frac{L^2}{4}(R-\frac{2L}{3})$;
		\item $\M$ is the image of a ReLU deep neural network, as discussed in Corollary~\ref {cor:ReLu}, of dimension $M\lesssim \frac{L^2}{4}(R-\frac{2L}{3})$.
	\end{itemize}
\end{corollary}
An analogous statement can be derived for the manifold case.

\section{Outlook} \label{sec:outlook}

The main contribution of this work is demonstrating that the mapping from a tuple of matrices to their corresponding Gram matrices is bi-Lipschitz, significantly advancing recent developments in the field. However, several important questions remain open for future investigation. Below, we outline a few of them.

\begin{itemize}
	\item
	\textbf{X-ray crystallography.} In X-ray crystallography, the objective is to recover a sparse signal---sparse in the standard basis---from its Fourier magnitudes~\cite{elser2018benchmark}. Despite the significance of this application, its mathematical foundations have only recently been explored in depth~\cite{bendory2020toward,ghosh2023sparse}. Notably,~\cite{bendory2020toward} conjectured---and proved in a limited parameter regime---that a sparse signal can be uniquely determined from its Fourier magnitudes, provided the number of non-zero entries does not exceed half the signal's dimensionality.  
	A key challenge is to understand when this mapping is also robust, as noise is inherent to this imaging modality. Importantly, while the results of this paper hold for sparsity under a generic basis, they do not necessarily apply to the standard basis. Proving this result specifically for the standard basis is a critical open direction.
	
	\item  \textbf{Lipschitz constants.} 
	Definition~\ref{def:biLipschitz} of bi-Lipschitzness is relatively weak, as it only requires the constants \( C_1 \) and \( C_2 \) to be nonzero and finite. However, in practical applications, these constants must be tightly controlled. Otherwise, even small measurement errors can result in significant estimation errors. Therefore, it is important to ensure that these constants are kept small and well-bounded. In this paper, we have a uniform bound of $\sqrt{2}$ for the upper Lipschitz constant, but the lower Lipschitz bound for a given prior set can, {\em a priori}, be arbitrarily close to zero. An important question for future work is to consider
	priors such as linear subspaces spanned by random (with respect to a fixed probability distribution) vectors and, for these priors, give probabilistic estimates for the lower
	Lipschitz bound $C_1$. This problem is a natural extension of earlier work of Eldar and Mendelson in frame phase retrieval~\cite{eldar2014phase}. 
	
	\item \textbf{Non-uniform distribution over the group.}
	The generalized phase retrieval problem arises from the MRA model~\eqref{eq:mra} under the assumption of a uniform distribution over the group \( G \). However, in many practical scenarios---most notably in cryo-EM---the distribution over \( G \), $\rho(G)$, is typically non-uniform~\cite{tan2017addressing,lyumkis2019challenges}. In such cases, the second moment takes the form  
	$\int_G \rho(g)(g \cdot x)(g \cdot x)^* \, dg,$  
	which defines a rich class of invariant functions of total degree three on \( R(G) \times V \), where \( R(G) \) denotes the regular representation of \( G \).  
	The core challenge lies in the fact that directly estimating all such invariants from the MRA observations is infeasible. Consequently, it becomes necessary to extend the injectivity results of~\cite{bendory2024transversality} and the bi-Lipschitz results established in this paper.
	
\end{itemize}

\section*{Acknowledgment}
T.B. and D.E. were supported by BSF grant 2020159. T.B. was also partially supported by NSF-BSF grant 2019752, ISF grant 1924/21, and by The Center for AI and Data
Science at Tel Aviv University (TAD). N.D. was supported by ISF grant 272/23. D.E. was also supported by NSF-DMS 2205626.
D.E. is grateful to Joey Iverson and Tanya Christiansen for helpful discussions.

\bibliographystyle{plain}

\end{document}